\documentclass[a4paper, 10 pts]{article}
\usepackage{fullpage}
\usepackage{microtype}%if unwanted, comment out or use option "draft"
\usepackage{amsmath, mathtools}
\usepackage{amssymb}
\usepackage{amsthm}
\usepackage{upgreek}
\usepackage[normalem]{ulem}
\usepackage{tikz}
\usetikzlibrary{decorations.pathreplacing}

\usepackage{microtype}
\usepackage{authblk}
\usepackage[colorinlistoftodos]{todonotes}
\usepackage{enumitem}
\usepackage{appendix}
\usepackage{graphicx}
\DeclareGraphicsExtensions{.pdf}
\usepackage{algorithm}
\usepackage[noend]{algpseudocode}

\newtheorem{theorem}{Theorem}
\newtheorem{assumption}[theorem]{Assumption}
\newtheorem{proposition}[theorem]{Proposition}
\newtheorem{claim}[theorem]{Claim}
\newtheorem{conjecture}{Conjecture}
\newtheorem{corollary}{Corollary}
\newtheorem{lemma}{Lemma}

\newcommand{\polylog}{\mathrm{polylog}}

\newcommand{\GPW}{\mathsf{GPW}}
\newcommand{\col}{\mathsf{col}}
\newcommand{\head}{\mathsf{head}}
\newcommand{\tail}{\mathsf{tail}}
\newcommand{\pred}{\mathsf{pred}}
\newcommand{\cQ}{\mathcal{Q}}
\newcommand{\eps}{\varepsilon}
\newcommand{\event}{\mathcal{E}}
\newcommand{\E}{\mathbb{E}}

\title{The zero-error randomized query complexity of the pointer function}

\author{Jaikumar Radhakrishnan \hspace{1cm}
       Swagato Sanyal\\
    Tata Institute
  of Fundamental Research, Mumbai\\ \texttt{\{jaikumar, swagato.sanyal\}@tifr.res.in}}
 
\date{}
\begin{document}

\maketitle

\begin{abstract}
The pointer function of G{\"{o}}{\"{o}}s, Pitassi and Watson
\cite{DBLP:journals/eccc/GoosP015a} and its variants have recently
been used to prove separation results among various measures of
complexity such as deterministic, randomized and quantum query
complexities, exact and approximate polynomial degrees, etc. In
particular, the widest possible (quadratic) separations between
deterministic and zero-error randomized query complexity, as
well as between bounded-error and zero-error randomized query
complexity, have been obtained by considering {\em
  variants}~\cite{DBLP:journals/corr/AmbainisBBL15} of this
pointer function.

However, as was pointed out in
\cite{DBLP:journals/corr/AmbainisBBL15}, the precise zero-error
complexity of the original pointer function was not known.  We show a
lower bound of $\widetilde{\Omega}(n^{3/4})$ on the zero-error
randomized query complexity of the pointer function on $\Theta(n \log
n)$ bits; since an $\widetilde{O}(n^{3/4})$ upper bound is already
known \cite{DBLP:conf/fsttcs/MukhopadhyayS15}, our lower bound is
optimal up to a factor of $\polylog\, n$.
\end{abstract}

%\graphicspath{/home/swagato/Desktop/GPW_revised/}

%The model of decision trees is one of the simplest models of computation. In this model, an algorithm for computing a Boolean function is given query access to the input.
%The algorithm queries different bits of the input, possibly in an adaptive fashion, and eventually outputs a bit. The objective is to minimize the number of queries made.
%The amount of computation is generally not the quantity of interest in this model.

%{\bf

%--- What did Ambainis {\em et al.} use to lower bound $R_0$?

%--- What sort of separation between certificate complexity and
%zero-error randomized complexity is known? What sort of techniques are
%used to show them?
%}

\section{Introduction}
%\graphicspath{/home/swagato/Desktop/GPW_revised/}
\label{intro}
%The model of decision trees is one of the simplest models of computation. In this model, an algorithm for computing a Boolean function is given query access to the input.
%The algorithm queries different bits of the input, possibly in an adaptive fashion, and eventually outputs a bit. The objective is to minimize the number of queries made.
%The amount of computation is generally not the quantity of interest in this model.

Understanding the relative power of various models of computation is a
central goal in complexity theory. In this paper, we focus on one of
the simplest models for computing boolean functions---the query model
or the decision tree model. In this model, the algorithm is required
to determine the value of a boolean function by querying individual
bits of the input, possibly adaptively. The computational resource we seek to
minimize is the number of queries for the worst-case input.  That is,
the algorithm is charged each time it queries an input bit, but not for its
internal computation.

There are several variants
of the query model, depending on whether \sout{or not} randomization
is allowed, and on whether error is acceptable. Let
$D(f)$ denote the deterministic query complexity of $f$, that is,
the maximum number of queries made by the algorithm for the worst-case
input; let $R(f)$ denote the maximum number of queries made by the
best randomized algorithm that errs with probability at most $1/3$
(say) on the worst-case input. Let $R_0(f)$ be the zero-error
randomized query complexity of $f$, that is, the expected number of
queries made for the worst-case input by the best randomized algorithm for $f$
that answers correctly on every input.

The relationships between these query complexity measures have been
extensively studied in the literature. That randomization can lead to
significant savings has been known for a long time.  Snir
\cite{DBLP:journals/tcs/Snir85} showed a $O(n^{\log_43})$ randomized
linear query algorithm (a more powerful model than what we discussed)
for complete binary NAND tree function for which the deterministic
linear query complexity is $\Omega(n)$. Later on Saks and
Wigderson~\cite{DBLP:conf/focs/SaksW86} determined the zero-error
randomized query complexity of the complete binary NAND tree function
to be $\Theta(n^{0.7536\dots})$. They also presented a result of Ravi
Boppana which states that the uniform rooted ternary majority tree function
has randomized zero-error query complexity $O(n^{0.893\dots})$ and
deterministic
query complexity $n$. All these example showed that randomized query
complexity can be substantially lower than its deterministic
counterpart. On the other hand, Nisan showed that the $R(f) =
\Omega(D(f)^{1/3})$ \cite{DBLP:journals/siamcomp/Nisan91}.  Blum and
Impagliazzo \cite{DBLP:conf/focs/BlumI87}, Tardos
\cite{DBLP:journals/combinatorica/Tardos89}, Hartmanis and Hemachandra
\cite{DBLP:conf/coco/HartmanisH87} independently showed that
$R_0(f)=\Omega(D(f)^{1/2})$. Thus, the question of the largest
separation between deterministic and randomized complexity remained
open. Indeed, Saks and Wigderson conjectured that the complete binary
NAND tree function exhibits the widest separation possible between
these two measures of complexity.
\begin{conjecture}[\cite{DBLP:conf/focs/SaksW86}]
\label{CONJ:SW}
For any boolean function $f$ on $n$ variables, $R_0(f) =
\Omega(D(f)^{0.753\dots})$.
\end{conjecture}

%In their $1986$ paper, Saks and Wigderson \cite{DBLP:conf/focs/SaksW86} gave examples of recursive NAND trees and recursive MAJORITY trees, for which
%they credited Snir and Ravi Bopanna respectively. In both these functions, the deterministic and randomized zero-error query complexity are polynomially separated. In the same
%paper, Saks and Wigderson studied binary uniform NAND trees, and showed that $R_0(F) = \Theta(D(F)^{0.753..})$ where $F$ is the binary uniform NAND tree function. 

This conjecture was recently refuted independently by Ambainis {\em et
  al.}~\cite{DBLP:journals/corr/AmbainisBBL15} and Mukhopadhyay and
Sanyal~\cite{DBLP:conf/fsttcs/MukhopadhyayS15}. Both works based their
result on the pointer function introduced by G{\"{o}}{\"{o}}s, Pitassi
and Watson \cite{DBLP:journals/eccc/GoosP015a}, who used this function
to show a separation between deterministic decision tree complexity
and unambiguous non-deterministic decision tree complexity. In
Section~\ref{gpw}, we present the formal definition of the function
$\GPW^{r \times s}$, which is a Boolean function on
$\widetilde{\Theta}(rs)$ bits.

Mukhopadhyay and Sanyal~\cite{DBLP:conf/fsttcs/MukhopadhyayS15} used
$\GPW^{s \times s}$ to obtain the following refutation of
Conjecture~\ref{CONJ:SW}: $R_0(\GPW^{s \times s})=
\widetilde{O}(s^{1.5})$ while $D(\GPW^{s \times s}) =
\Omega(s^2)$. While this shows that $\GPW^{s \times s}$ witnesses a
wider separation between deterministic and zero-error randomized query
complexities than conjectured, the separation shown is not the widest
possible for a Boolean function.  Independently, Ambainis {\em et al.}
modified $\GPW^{s \times s}$ in subtle ways, to establish the widest
possible (near-quadratic) separation between deterministic and
zero-error randomized query complexity, and between zero-error
randomized and bounded-error randomized query complexities.

Ambainis et al. \cite{DBLP:journals/corr/AmbainisBBL15} pointed out, however, that the precise zero-error
randomized query complexity (i.e. $R_0(\GPW^{s \times s})$)
was not known.  One could ask if the optimal separation demonstrated by Ambainis {\em
  et al.} is also witnessed by $\GPW^{s \times s}$ itself. In this
work, we prove a near-optimal lower bound on the zero-error randomized
query complexity of $\GPW^{r \times s}$, which is slightly
more general than the $\GPW^{s \times s}$ considered
in earlier works.
\begin{theorem}[Main theorem]
\label{mainthm}
$R_0(\GPW^{r \times s}) = \widetilde{\Omega}(r+\sqrt{r}s)$.
\end{theorem}
Such a result essentially
claims that randomized algorithms cannot efficiently locate
certificates for the function. This would be true, for example, if the
function could be shown to require large certificates, since the
certificate complexity of a function is clearly a lower bound on its
zero-error randomized complexity. This straightforward approach does
not yield our lower bound, as the certificate complexity of $\GPW^{r
  \times s}$ is $\widetilde{O}(r+s)$. In our proof, we set up a
special distribution on inputs, and by analyzing the expansion
properties of the pointers, show that a certificate will evade a
randomized algorithm that makes only a small number of queries.  In fact, the distribution we
devise is almost entirely supported on
inputs $X$ for which $\GPW^{r \times s}(X)=0$. This is not an accident:
a randomized algorithm can quickly find a certificate for inputs
$X$ if $\GPW^{r \times s}(X)=1$ (see Theorem~\ref{ms2} below).

It
follows from Theorem~\ref{mainthm} that the algorithm of Mukhopadhyay
and Sanyal \cite{DBLP:conf/fsttcs/MukhopadhyayS15} is optimal up to polylog factors.
\begin{corollary}
$R_0(\GPW^{s \times s}) = \widetilde{\Omega}(s^{1.5})$.
\end{corollary}

In addition to nearly determining the zero-error complexity of the 
original $\GPW^{s\times s}$ function, our result has two interesting 
consequences.
\begin{enumerate}
 \item[(a)] The above mentioned result of Mukhopadhyay and
   Sanyal~\cite{DBLP:conf/fsttcs/MukhopadhyayS15} showed that
   $R_0(\GPW^{s \times s})=\widetilde{\Omega}(D(\GPW^{s \times
     s})^{0.75})$. Our main theorem shows that $\GPW^{s \times s}$
   cannot be used to show a significantly better separation between
   the deterministic and randomized zero-error complexities (ignoring
   $\polylog$ factors). However, the function $\GPW^{s^2
     \times s}$ allows us to derive a better separation\footnote{In
     \cite{WEB:SAaranson}, a similar separation between $R(\GPW^{s^2
       \times s})$ and $D(\GPW^{s^2 \times s})$ is mentioned.}:
   $R_0(\GPW^{s^2 \times s}) = O(D(\GPW^{s^2 \times s})^{2/3})$.  Our
   main theorem shows that this is essentially the best separation
   that can be derived from $\GPW^{r \times s}$ by varying $r$
   relative to $s$, so this method cannot match the near-quadratic
   separation between these measures shown by Ambainis {\em et
     al.}~\cite{DBLP:journals/corr/AmbainisBBL15} by considering a
   variant of the $\GPW^{s \times s}$ function.

\item[(b)]  $\GPW^{s \times s}$ exposes a non-trivial polynomial
  separation between the zero-error and bounded-error
  randomized query complexities: $R(\GPW^{s \times
    s})=\widetilde{O}(R_0(\GPW^{s \times s})^{2/3})$. This falls
    short of the near-quadratic separation shown by
    Ambainis {\em et al.}~\cite{DBLP:journals/corr/AmbainisBBL15},
    but note that before that result no
    separation between these measures was known.
\end{enumerate}

\section{The $\GPW$ function}
\label{gpw}

\newcommand{\ptr}{\mathsf{ptr}} \newcommand{\cA}{\mathcal{A}}
\newcommand{\hc}{\hat{c}} The input $X$ to the \emph{pointer
  function}, $\GPW^{r\times s}$, is arranged in an array with $r$ rows
and $s$ columns. The cell $X[i,j]$ of the array contains two pieces of
data, a bit $b_{ij} \in \{0,1\}$ and a pointer $\ptr_{ij} \in
([r]\times [s])\cup \{ \bot \}$.
%If $\ptr_{ij} = (k, \ell) \in [r]
%\times [s]$, we define $\ptr_{ij}[1] =k$ and $\ptr_{ij}[2]= \ell$. See Figure~\ref{ip}.

Let
$\cA$ denote the set of all such arrays. The function $\GPW^{r \times
  s}: \cA \rightarrow \{0,1\}$ is defined as follows: $\GPW^{r \times
  s}(X)=1$ if and only if the following three conditions are
satisfied.
\begin{enumerate}
\item There is a unique column $j^*$ such that for all rows $i \in
  [r]$, we have $b_{ij^*}=1$.

\item In this column $j^{*}$, there is a unique row $i^*$ such that
  $\ptr_{i^*j^*} \neq \bot$.

\item Now, consider the sequence of locations $(p_k: k=0, 1, \ldots,
  s-1)$, defined as follows: let $p_0 = (i^*, j^*)$, and for $k=0, 1,
  \ldots, s-2$, let $p_{k+1} = \ptr_{p_k}$. Then, $p_0, p_1, \ldots,
  p_{s-1}$ lie in distinct columns of $X$, and $b_{p_k} = 0$ for
  $k=1,2,\ldots, s-1$. In other words, there is a {\em chain of
    pointers}, which starts from the unique location in column $j^*$
  with a non-null pointer, visits all other columns in exactly $s-1$
  steps, and finds a $0$ in each location it visits (except the
  first).
\end{enumerate}
Note that $\GPW^{r \times s}$ can be thought of as a Boolean function on $\Theta(rs \log rs)$ bits.

\paragraph{Upper Bound}

The pointer function $\GPW^{r \times s}$, as defined above, is
parameterized by two parameters, $r$ and $s$.  G{\"{o}}{\"{o}}s,
Pitassi and Watson~\cite{DBLP:journals/eccc/GoosP015a} focus on the
special case where $r=s$.  Mukhopadhyay and
Sanyal~\cite{DBLP:conf/fsttcs/MukhopadhyayS15} also state their
zero-error randomized algorithm with $\widetilde{O}(s^{1.5})$ queries
for this special case; however, it is straightforward to extend their
algorithm so that it applies to the function $\GPW^{r \times s}$.
\begin{theorem}
\label{ms1}
$R_0(\GPW^{r \times s}) = \widetilde{O}(r+\sqrt{r}s)$.
\end{theorem}
Mukhopadhyay and Sanyal also gave a one-sided error randomized query
algorithm that makes $\widetilde{O}(s)$ queries on average but never
errs on inputs $X$, where $\GPW^{s \times s}(X)=1$. Again a straightforward
extension yields the following.
\begin{theorem}
\label{ms2}
There is a randomized query algorithm that makes $\widetilde{O}(r+s)$ queries
on each input, computes $\GPW^{r \times s}$ on each input with probability at
least $1/3$, and in addition never errs on inputs $X$ where $\GPW^{r \times s}(X)=1$.
\end{theorem}
Theorem~\ref{mainthm}, thus, completely determines the deterministic
and all randomized query complexities of a more general function
$\GPW^{r \times s}$.
\begin{figure}
\begin{tikzpicture}[align=center]

\draw[red] (0,0) grid[step=1cm] (5,5);

\draw[teal] (2,3) -- (3,2);
\node at (2.5,2.25){$1$};
\draw[->, ultra thick] (2.75,2.75) to [out=0, in=+90](4.5,1);
%\node at (2.75,2.75){$\bot$};

\foreach \x in {1,2,3,4,5}
	\draw[teal] (0,\x) -- (\x,0);
\foreach \x in {1,2,3,4,5}
	\draw[teal] (\x,5) -- (5,\x);
\foreach \x in {1,4,5}
	\foreach \y in {2,3,4}
		\node at (\x-0.5,\y-0.75){$0$};
\foreach \x in {2,3,5}
	\foreach \y in {1,5}
		\node at (\x-0.5,\y-0.75){$0$};
\foreach \x in {1,4}
	\foreach \y in {1,5}
		\node at (\x-0.5,\y-0.75){$1$};
\node at (3-0.5,3-0.75){$1$};
\node at (2-0.5,4-0.75){$1$};
\node at (2-0.5,3-0.75){$1$};
\node at (3-0.5,4-0.75){$0$};
\node at (2-0.5, 2-0.75){$0$};
\node at (3-0.5,2-0.75){$1$};

\draw[->, ultra thick] (0.75, 4.75) to [out=0, in=+120] (2.2,4);
\draw[->, ultra thick] (2.75, 3.75) to  (3,3.75);
\draw[->, ultra thick] (3.75, 3.5) to [out=170, in=+90] (0.25,2);
\draw[->, ultra thick] (0.75, 1.75) to [out=0, in=-90] (1.5,2);
\draw[->, ultra thick] (1.65, 1.65) to [out=160, in=-90] (0.15,4);
\draw[->, ultra thick] (0.80,0.50) to [out=0, in=-90] (2.25,1);

\node at (0.75,2.75){$\bot$};
\node at (1.75,4.75){$\bot$};
\node at (2.75,4.75){$\bot$};
\node at (3.75,4.75){$\bot$};
\node at (4.75,4.75){$\bot$};
\node at (0.75,3.75){$\bot$};
\node at (1.75,3.75){$\bot$};
\node at (4.75,3.75){$\bot$};
\node at (1.75,2.75){$\bot$};
\node at (3.75,2.75){$\bot$};
\node at (4.75,2.75){$\bot$};
\node at (2.75,1.75){$\bot$};
\node at (3.75,1.75){$\bot$};
\node at (4.75,1.75){$\bot$};
\node at (1.75,0.75){$\bot$};
\node at (2.75,0.75){$\bot$};
\node at (3.75,0.75){$\bot$};
\node at (4.75,0.75){$\bot$};

\end{tikzpicture}

\caption{Input to $\GPW^{r \times s}$ for $r=5,s=5$.}\label{ip}
\end{figure}
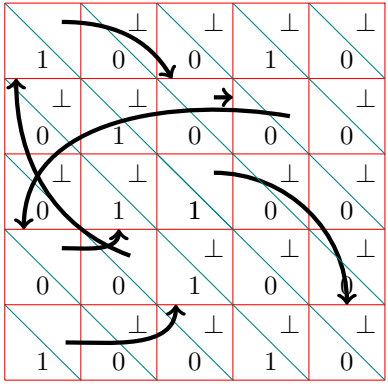

\subsection{The distribution} 
\label{sec:distr}
To show our lower bound, we will set up a distribution on inputs in
$\cA$. Let $V$ be the locations in the first $s/2$ columns, i.e., $V=
[r] \times [s/2]$; let $W$ be the locations in the last $s/2$ columns,
i.e., $W=[r] \times ([s] \setminus [s/2])$.  
In order to describe the random input $X$, we will need the following 
definitions.

\paragraph{Pointer chain:} For an input in $\cA$, we say that a sequence of
locations $\mathbf{p}=\langle \ell_0, \ell_1, \ell_2, \ldots, \ell_k \rangle$
is a pointer chain, if for $i=0,1, \ldots, k-1$, $\ptr_{\ell_i}=
\ell_{i+1}$; the location $\ell_0$ is the \emph{head} of the
$\mathbf{p}$ and is denoted by $\mathsf{head}(\mathbf{p})$; similarly,
$\ell_k$ is the \emph{tail} of $\mathbf{p}$ and is denoted by
$\mathsf{tail}(\mathbf{p})$.  Note that $\ptr(\ell_k)$ is not
specified as part of the definition of pointer chain $\mathbf{p}$; in
particular, it is allowed to be $\bot$.

\paragraph{Random pointer chain:} To build our random input $X$, we will
assign the pointer values of the various cells of $X$ randomly so that
they form appropriate pointer chains. For a set of locations $S$ we
build a {\em random pointer chain on $S$} as follows. First, we
uniformly pick a permutation of $S$, say $\langle \ell_0, \ell_2,
\ldots, \ell_k \rangle$. Then, we set $\ptr_{\ell_i}=\ell_{i+1}$ (for
$i=0,1,\ldots,k-1$). We will make such random assignments for sets
$S$ consisting of consecutive locations in some row of $W$. We call
the special (deterministic) chain that starts at the first (leftmost)
location of $S$, visits the next, and so on, until the last
(rightmost) location, a {\em path.}  Given two pointer chains
$\mathbf{p}_1$ and $\mathbf{p}_2$ on disjoint sets of locations
$S_1$ and $S_2$, we may
set $\ptr_{\tail(\mathbf{p}_1)} = \head(\mathbf{p}_2)$, and obtain
a single pointer chain
on $S_a \cup S_b$, whose head is $\head(\mathbf{p}_1)$ and tail is
$\tail(\mathbf{p}_2)$. We will refer to this operation as the {\em
  concatenation} of $\mathbf{p}_1$ and $\mathbf{p}_2$.

We are now ready to define the random input $X$. First, consider
$W$. For all $\ell \in W$, we set $b_{\ell}=0$. To describe the
pointers corresponding to $W$, we partition the columns of $W$ into
$K:=\log s - 3 \log \log s$ {\em blocks}, $W_1,\ldots,W_K$, where
$W_1$ consists of the first $s/(2K)$ columns of $W$, $W_2$ consists of
the next $s/(2K)$ columns, and so on. 
\newenvironment{centermath}
 {\begin{center}$\displaystyle}
 {$\end{center}}
 \begin{centermath}
\left[
\begin{array}{ c c c | c c c  c c c}
&&&\\
&&&\\
&\Large \mbox{$V$}&  & \hspace{.3in} & \large \mbox{$W_1$}  & \large \mbox{$W_2$} \hspace{0.3in }\mbox{$\ldots$} \hspace{0.3in } \large \mbox{$W_K$} \hspace{.3in} \\
&&& \\
&&& \\
\end{array}
\right]
\end{centermath}
The block $W_j$, will be further divided into {\em bands;} however,
the number of bands in different $W_j$ will be different.  There will
be $20 \cdot 2^j$ bands in $W_j$, each consisting of $w_j := s/(20
\cdot 2^{j} \cdot 2K)$ contiguously chosen columns. See
Figure~\ref{block}.  

\begin{figure}[h]
\begin{tikzpicture}
\draw[thick](0,0) rectangle (10,4);
\node at (10.5,2){$\ldots$};
\draw[thick](11,0) rectangle (13,4);

\foreach \x in {2,4,6,8}
	\draw (\x,0) -- (\x,4);
\draw(0,3) -- (2,3);
\draw(0,2.3) -- (2,2.3);
\foreach \x in {0,1/3,2/3,1,4/3,5/3}
	\draw(\x,2.3) -- (\x+0.1655,3);

\draw [thick, decorate,,decoration={mirror, brace,amplitude=10pt},xshift=-4pt,yshift=0pt]
(0.15,2.1) -- (2.15,2.1);

\node at (1,1.5){Segment};
\node(band) at (3,-1){\Large Bands};
\draw[->] (band) -- (3.3,0.2);
\draw[->] (band) -- (4.4,0.2);
\draw[->] (band) -- (6.2,0.5);

\node at (5,4.5){width $(w_j)=20 \cdot 2^j$};

\draw[<->] (4,4.2) -- (6,4.2);

\draw[<->] (-1,4) -- (-1,0);
\node at (-1.5,2){$r$};
\draw[<->] (0,5) -- (13,5);
\node at (5,5.4){\Large $\frac{s/2}{\log s - 3\log \log s}$};

\node at (5,-1.5){\huge $W_j$};

\end{tikzpicture}
\caption{Bands and segments inside block $W_j$.} \label{block}
\end{figure}
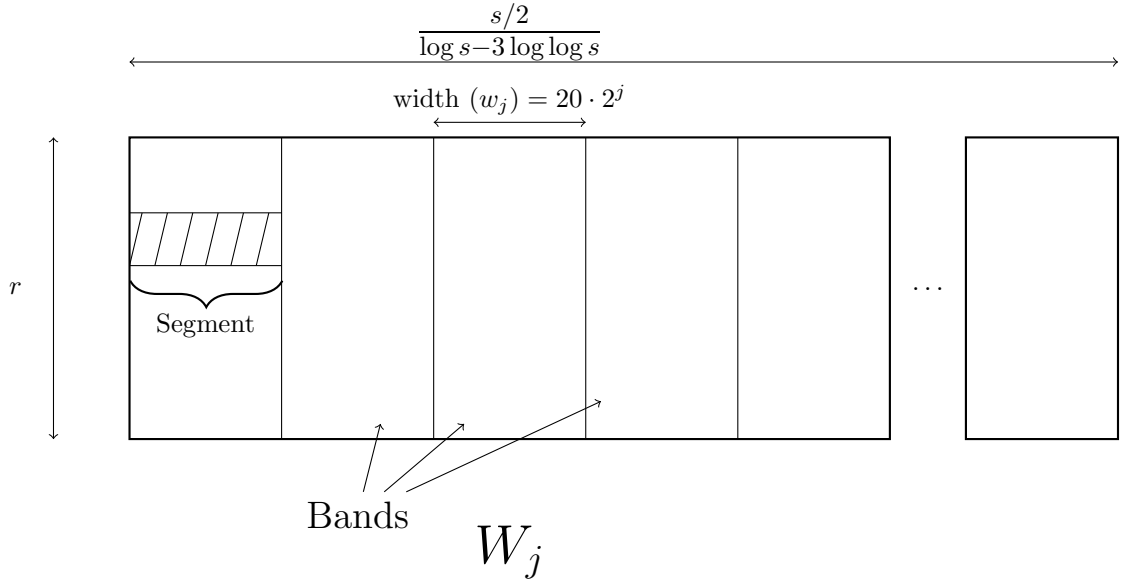

Each such band will have $r$ rows; the locations
in a single row of a band will be called a {\em segment}; we will
divide each segment into two equal parts, left and right, each with
$w_j/2$ columns. (See Figure~\ref{seg}.)

We
are now ready to specify the pointers in each segment of $W_i$. In the
first half of each segment we place a random (uniformly chosen)
pointer chain; in the right half we place a path starting at its
leftmost cell and leading to its rightmost cell.  Once all pointer
chains in all the segments in a given row are in place, we concatenate
them from left to right.  All pointers in the last column of $W$ are
set to $\bot$. In the resulting input, each row of $W$ is a single
pointer chain with head in the leftmost segment of $W_1$ and tail in
the last column of $W$. This completes the description of $X$ for the
locations in $W$.

Next, we consider locations in $V$. Let $q:= 500 \log s / \sqrt{r}$.
Independently, for each location $\ell \in V$:
\begin{itemize}
\item with probability $q$, set $b_{\ell}=0$ and $\ptr_\ell$ to be a
  random location that is in the {\em left half} of some segment in
  $W$ (that is, among all locations that fall in the left half of some
  segment, pick one at random and set $\ptr_\ell$ to that location);
\item with probability $1-q$, set $b_{\ell}=1$ and $\ptr_\ell = \bot$.
\end{itemize}
This completes the description of the random input $X$.

\begin{figure}[h]
\begin{tikzpicture}
\draw[thick](0,0) rectangle (10,1);

\foreach \x in {0.5,1,1.5,2,2.5,3,3.5,4,4.5,5.5,6,6.5,7,7.5,8,8.5,9,9.5}
	\draw[dotted] (\x,0) -- (\x,1);

\draw[ultra thick] (5,0) -- (5,1);

\draw[thick] (3.9,-0.1) rectangle (4.6,1.1);
\draw[thick] (-0.1,-0.1) rectangle (0.6,1.1);

\draw[->,thick] (4.3,1) to [out=+90, in=+90](0.8,1);
\draw[->, thick] (0.8,0) to [out=-90, in=-90](2.2,0);

\draw[->, thick, rounded corners] (2.3,0) -- (2.3,-0.3) --(2.7,-0.3) -- (2.7,0);
\draw[->, thick, rounded corners] (2.8,0) -- (2.8,-0.3) --(3.3,-0.3) -- (3.3,0);
\draw[->, thick] (3.3,1) to [out=+90, in=+90](1.3,1);
\draw[->, thick] (1.3,0) to [out=-90, in=-90](1.8,0);
\draw[->, thick] (1.8,1) to [out=+90, in=+90](4.8,1);
\draw[->, thick] (4.8,0) to [out=-90, in=-90](3.7,0);
\draw[->, thick] (3.6,0) to [out=-90, in=-90](0.3,0);

\draw[->, ultra thick] (0,-0.75) -- (0.2,0.4);
\node at (0,-1){\textbf{tail}};

\draw[->, ultra thick] (4,-0.75) -- (4.2,0.4);
\node at (4,-1){\textbf{head}};

\draw[->, thick, rounded corners] (0.3,1) -- (0.3,3) --(5.25,3) -- (5.25,1);

\draw[->, thick] (5.25,0.5) --(5.75,0.5);

\foreach \x in {5.75,6.25,6.75,7.25,7.75,8.25,8.75,9.25}
	\draw[->,thick](\x, 0.5) -- (\x+0.5,0.5);
\draw[->,dotted](9.75, 0.5) -- (10.25, 0.5);

\draw [thick, decorate,,decoration={mirror, brace,amplitude=10pt},xshift=-4pt,yshift=0pt]
(0.1,-1.5) -- (5.05,-1.5);

\draw [thick, decorate,,decoration={mirror, brace,amplitude=10pt},xshift=-4pt,yshift=0pt]
(5.15,-1.5) -- (10.1,-1.5);

%\node at (2.5,-3){Left half};
\node at (2.5,-2.1){random pointer chain};

%\node at (7.5,-3){Right half};
\node at (7.5,-2.1){path};

\end{tikzpicture}
\caption{A segment consists of a random pointer chain concatenated with a path.} \label{seg}
\end{figure}
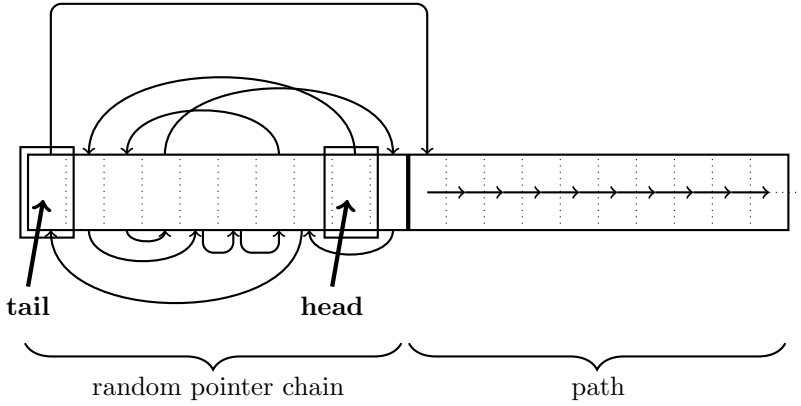

\section{The lower bound for $\GPW^{r \times s}$}
We will consider algorithms that are given query access to the input bits
of $\GPW^{r \times s}$. A location $\ell \in [r] \times [s]$ of an
input $X \in \cA$ is said to be queried if either $b_\ell$ is queried,
or some bit in the encoding of $\ptr_\ell$ is queried. By \emph{number of
queries}, we will always mean the number of locations queried.  A lower
bound on the number of locations queried is clearly a lower bound on
the number of bits queried.

%\paragraph{Idea}
It can be shown that the certificate complexity of $\GPW^{r \times s}$
is $\Omega(r+s)$; hence $R_0(\GPW^{r \times s})=\Omega(r+s)$. It remains
to show that any zero-error randomized query
algorithm for $\GPW^{r \times s}$ must make
$\Omega(\sqrt{r} s/\polylog(s))$ queries in expectation. We will assume
that there is a significantly more efficient algorithm and derive a
contradiction.
\begin{assumption}
\label{asmp}
There is a zero-error randomized algorithm that makes at most
$\sqrt{r}{s}/(\log s)^5$ queries in expectation (taken over the algorithm's
coin tosses) on every input $X$.
\end{assumption}
If $r <
(\log s)^3$ (say), then this assumption immediately leads to a contradiction
because $R_0(\GPW^{r \times s})=\Omega(s)$.
So, we will assume that $r \geq (\log s)^3$.

Consider inputs $X$ drawn according to the distribution described in
the previous section.  Since with probability $1-o(1)$ every column of
$X$ has at least one zero (see Lemma~\ref{1col} (a)), $\GPW^{r\times
  s}(X)=0$ with probability $1-o(1)$; thus, the algorithm returns the
answer $0$ with probability $1-o(1)$.  Taking expectation over inputs
$X$ and the algorithm's coin tosses, the expected number of queries
made by the algorithm is at most $\sqrt{r}s/(\log s)^5$. Using
Markov's inequality, with probability $1-o(1)$, the algorithm stops
after making at most $\sqrt{r}s/(\log s)^4$ queries. By truncating the long
runs and fixing the random coin tosses of the
algorithm, we obtain a deterministic algorithm. Hence we have the following.
\begin{proposition} \label{prop:deterministic}
If Assumption~\ref{asmp} holds, then there is a deterministic
algorithm that (i) queries at most $\sqrt{r} s/(\log s)^4$ locations,
(ii) never returns a wrong answer (it might give no answer on some
inputs), and (iii) returns the answer $0$ with probability $1-o(1)$
for the random input $X$.
\end{proposition}
Fix such a deterministic query algorithm $\cQ$. We will show that with
high probability the locations of $X$ that are left unqueried by $\cQ$
can be modified to yield an input $X'$ such that $\GPW^{r \times
  s}(X')=1$. Thus, with high probability, $\cQ(X') = \cQ(X) =0$.  This
contradicts
Proposition~\ref{prop:deterministic} (ii). In fact, in the next section,
we formally establish the following.
\begin{lemma}[Stitching lemma] \label{lm:stitching}
With probability $1-o(1)$ over the choices of $X$, there is an input
$X'\in \cA$ that differs from $X$ only in locations not probed by
$\cQ$ such that $\GPW^{r \times s}(X') = 1$.
\end{lemma}
By the discussion above, this immediately implies Theorem~\ref{mainthm}.

\section{The approach}

In this section, we will work with the algorithm $\cQ$ that is
guaranteed to exist by Proposition~\ref{prop:deterministic}.  For an
input $X \in \cA$ to $\GPW^{r \times s}$, let $G_X=(V', W',E)$ be a
bipartite graph, where $V'$ is the set of columns of $V$ and $W'$ is
the set of all bands in all blocks of of $W$. The edge set $E(G_X)$
is obtained as
follows.  Recall that pointers from $V$ lead to segments in $W$. Each
such segment contains a pointer chain. For a location $\ell$ in such a
chain, let $\pred(\ell)$ denote the location $\ell'$
that precedes $\ell$ in the chain (if $\ell$ is the head,
then $\pred(\ell)$ is undefined); thus, 
$\ptr_{\ell'}=\ell$.  We include the edge $(j,\beta)$ (connecting
column $j \in V'$ to band $\beta \in W'$) in $E(G_X)$ if the following
holds: \\ \ \\
\emph{There is a location $v$ in column $j$ and a segment $p$ in some
row of band $\beta$ such that
\begin{enumerate}
\item[(c1)] $\ptr_v \in p$, that is, $\ptr_v$ is non-null and points to a
  location in the left half of segment $p$; 
\item[(c2)] $\pred(\ptr_v)$ is well defined and is not probed by $\cQ$; 
\item[(c3)] $\cQ$ makes fewer than $|p|/4$ probes in segment $p$. (Note that
  this implies that there is a location in the right half of $p$ that
  is left unprobed by $\cQ$.)
\end{enumerate}}
In the next section, we will show the following.
\begin{lemma}[Matching lemma] \label{lm:matching}
With probability $1-o(1)$ over the choice of $X$, for every subset $R
\subseteq V'$ of at most $s/(\sqrt{r}(\log s)^4)$ columns, there is a
matching in $G_X$ that saturates $R$.
\end{lemma}
 In this section, we will show how Lemma~\ref{lm:matching} enables us
 to modify the input $X$ to obtain an input $X'$ for which $\GPW^{r
   \times s}(X')=1$, thereby establishing Lemma~\ref{lm:stitching}.
\begin{lemma}
\label{1col}
\begin{enumerate}
\item[(a)] With probability $1-o(1)$, each column $j$ of the input $X$ has
  a location $\ell$ such that $b_\ell=0$.

\item[(b)] With probability $1-o(1)$, there is a column $j \in [s/2]$ such
  that $\cQ$ does not read any location $\ell$ in column $j$ with
  $b_\ell=0$.
\end{enumerate}
\end{lemma}
\begin{proof}
\begin{enumerate}
\item[(a)]
All the bits in the columns in $[s] \setminus [s/2]$ are $0$. We show
that with high probability, each column in $V'$ has a $0$. The
probability that a particular column in $V'$ does not have any $0$
is $(1-500 \log s/\sqrt{r})^r \leq s^{-\Omega(\sqrt{r})}$.  Thus the
probability that there is a column $j \in V'$ which does not have
any $0$ is at most $(s/2) \cdot s^{-\Omega(\sqrt{r})} = o(1)$.

\item[(b)] Suppose $\cQ$ makes $t \leq s \sqrt{r}/(\log s)^4$
  queries. For $i=1,2,\ldots, t$, let $R_i$ be the indicator variable for
  the the event that in the $i$-th query, $\cQ$ reads a $0$ from $V$.
  Then, the expected number of $0$'s read by $\cQ$ in $V$
is (we assume that $\cQ$ does not read the same location twice)
\begin{align}
&\sum_{i=1}^q \mathbb{E}[R_i] \leq t \cdot 500 \log s/\sqrt{r} \leq 500 s / (\log s)^3. \nonumber
\end{align}
By Markov's inequality, with probability $1-o(1)$, the number number
of $0$'s read by $\cQ$ is less than $s/2$. It follows, that there is a
column in $V$ in which $\cQ$ has read no $0$.
\end{enumerate}
\end{proof}

\begin{proof}[Proof of Lemma~\ref{lm:stitching}]
Assume that the high probability events of Lemmas~\ref{lm:matching}
and~\ref{1col} hold. This happens with probability $1-o(1)$. We will
now describe a sequence of modifications to the input $X$ at locations
not queried by $\cQ$ to transform it into a input $X'$ such that
$\GPW^{r \times s}(X') = 1$. Let $j^* \in V'$ be the column in $V$
guaranteed by Lemma~\ref{1col} (b). Define $A_0 = \{\col_1, \ldots,
\col_N\} \subseteq V' \setminus \{j^*\}$ to be the set of columns in
$V' \setminus \{j^*\}$ that are not completely read by $\cQ$ (i.e. each
column in $A_0$ has a location unread by $\cQ$). Let $\ell_i$ be a
location in the column $\col_i$ that is unread by $\cQ$.  We first
make the following changes to $X$, with the aim of starting a pointer
chain at column $j^*$ that passes through $\col_1, \col_2, \ldots,
\col_N$.
\begin{enumerate}
 \item[(i)] For each unread location $\ell$ in the column $j^*$, set $b_{\ell}$ to $1$.
 \item[(ii)] Let $\ell^*$ be the first unread location of $j^{*}$. Set $\ptr_{\ell^{*}}$ to $\ell_1$.
\item[(iii)] For each location $\ell \neq \ell^*$ in column $j^*$, set $\ptr_\ell$ to $\bot$.
 \item[(iv)] For $i=1, \ldots, N-1$, set $b_{\ell_i}$ to $0$ and $\ptr_{\ell_i}$ to $\ell_{i+1}$.
 \item[(v)] Set $b_{\ell_N}$ to $0$.
\end{enumerate}
Clearly, the locations modified are not probed by $\cQ$. Notice that
the current input has the pointer chain $\mathbf{p}_0=(\ell^*, \ell_1,
\ldots, \ell_N)$ and the head $\ell^*$ of the chain lies in the
all-ones column $j^*$. Furthermore, all locations on the chain except
$\ell^*$ have $0$ as their bit. We now show how to further modify our
input and extend $\mathbf{p}$ and visit the remaining columns through
locations with $0$'s. The columns in $W$ are already neatly arranged
in pointer chains. The difficulty is in ensuring that we also visit
the set of columns in $V'$ that are completely read by $\cQ$, for we
are not allowed to make any modifications there.  Let $A_1$ denote
these completely read columns in $V'$. Since $\cQ$ makes at most
$\sqrt{r}s/(\log s)^4$ queries, we have that $|A_1| \leq
s/(\sqrt{r}(\log s)^4)$. Lemma~\ref{lm:matching} implies that there
exists a matching $\mathcal{M}$ in $G_X$ that saturates $A_1$. Order
the elements of $A_1$ as $d_1, \ldots, d_L$ in such a way for all
$i=1, \ldots, L-1$, $\mathcal{M}(d_i) < \mathcal{M}(d_{i+1})$ (where we order
the bands in $W$ from left to right), that is, the band that is
matched with $d_i$ lies to the left of the band that is matched to
$d_{i+1}$.

We will now proceed as follows. For $i=1, \ldots, L$, we modify the
input (at locations not read by $\cQ$) appropriately to induce a
pointer chain $\mathbf{p}_i$. This pointer chain in addition to
visiting a contiguous set of columns in $W$, will visit column
$d_i$. By concatenating these pointer chains
  in order with the initial pointer chain $\mathbf{p}_0$
we obtain the promised input $X'$ for which $\GPW^{r\times
    s}(X')=1$.

To implement this strategy, recall that there is an edge in $G_X$
between the column $d_i$ and the band $\mathcal{M}(d_i)$. From the
definition of $G_X$, it follows that there is a location $q_i$ in
$d_i$ and a segment $S_i$ in band $\mathcal{M}(d_i)$ such that
\begin{enumerate}
\item[(s1)] $\ptr_{q_i}$ leads to the left half of $S_i$;
\item[(s2)] $\pred(\ptr_{q_i})$ is not probed by $\cQ$;
\item[(s3)] $\cQ$ makes fewer than $|S_i|/4$ queries in segment $S_i$.
\end{enumerate}
First, let us describe how $\mathbf{p}_1$ is constructed.  Let $a_1 =
\ptr_{q_1}$ and $b_1 = \pred(a_1)$ (by (s2) $b_1$ is not probed by
$\cQ$); let $c_1$ be the first location in the second half of $S_1$
that is not probed by $\cQ$ (by (s3) there is such a location). Now,
we modify the input $X$ by setting $\ptr_{b_1} = q_1$.  Then,
$\mathbf{p}_1$ is the pointer chain that starts at the head of
the leftmost
segment of $W_1$ in the same row as $S_1$ and continues until location
$c_1$. That is, starting from its head, it follows the pointers of the input until
$b_1$. Then it follows the pointer leading out of $b_1$  into $q_1$,
thereby visiting column $d_1$. After that, it follows the pointer out of $q_1$ and
comes to $a_1$, and keeps following the pointers until $c_1$. 

In general, suppose $\mathbf{p}_1, \mathbf{p}_2, \ldots,
\mathbf{p}_{i-1}$ have been constructed. Suppose
$\tail(\mathbf{p}_{i-1})$ appears in column $k_{i-1}$. Then,
$\mathbf{p}_i$ is obtained as follows. Let $a_i = \ptr_{q_i}$ and $b_i
= \pred(a_i)$; let $c_i$ be the first location in the second half of
$S_i$ that is not probed by $\cQ$. We modify the input by setting
$\ptr_{b_i} = q_i$. Then $\mathbf{p}_i$ is the pointer chain with its
head in the same row as $a_i$ and in column $k_{i-1}+1$; this pointer
chain terminates in location $c_i$.  See Figure~\ref{stitch}. Note that $\mathbf{p}_{i}$
entirely keeps to one row (the row of $S_i$), except for the diversion from $b_i$ to
$q_i$, when it visits column $d_i$ and returns to $a_i$.
When $i=L$, we let the pointer chain continue until the last column of
$W$.  

In obtaining the pointer chains $\mathbf{p}_1, \mathbf{p}_2,
\ldots, \mathbf{p}_L$, we modified $X$ at location $b_1, b_2, \ldots,
b_L$.  Finally, we concatenate the pointer chains $\mathbf{p}_0,
\mathbf{p}_1, \ldots, \mathbf{p}_L$; this requires us to modify $X$ at
locations $\ell_N=\tail(\mathbf{p}_0), c_1, c_2, \ldots, c_{L-1}$,
which were left unprobed by $\cQ$. The resulting input after these
modifications is $X'$.

The pointer chain obtained by this concatenation visits each column other
than $j^*$ exactly
once, and the bit at every location on it, other than its head, is
$0$. Hence, $\GPW^{r \times s}(X')=1$.
\end{proof}

\begin{figure}[h]
\begin{tikzpicture}

\draw (4,4) --(8,4);
\draw (4,3.5) -- (8,3.5);
\draw[ rounded corners] (4,4) -- (4.1,3.8) --(3.9,3.6) -- (4,3.5) ;
\draw[ rounded corners] (8,4) -- (8.1,3.8) --(7.9,3.6) -- (8,3.5);

\draw[draw=black, fill=gray, opacity=0.2] (6.5,3.5) rectangle (6.8,4);
\node(a) at (6.6,4.3){$c_{i-1}$};
\node(b) at (6.6,5.2){column $k_{i-1}$};
\draw[->] (b) -- (a);

\draw (6.8,2.5) -- (11.8,2.5);
\draw (6.8,3) -- (11.8,3);

\draw[ rounded corners] (11.8,3) -- (11.7,2.8) --(11.9,2.55) -- (11.8,2.5) ;
\draw (7.8,2.5) -- (7.8,3);
\draw (10.8,2.5) -- (10.8,3);
\draw (9.3,2.5) -- (9.3,3);

\draw[draw=black, fill=gray, opacity=0.2] (8.6,2.5) rectangle (8.8,3);
\draw[draw=black, fill=gray, opacity=0.2] (8.1,2.5) rectangle (8.3,3);
\draw[draw=black, fill=gray, opacity=0.2] (10.3,2.5) rectangle (10.5,3);
%\node at (8.9,2.75){$\cdot$};

\node at (8.7, 3.2){$a_i$};
\node at (8.2, 3.2){$b_i$};
%\draw[->, dotted,  rounded corners] (8.2,2.5) --  (8.3,2.25) -- (8.6,2.25) --  (8.65,2.5);
\draw[dotted,->] (8.3,2.75) -- (8.6,2.75);

\node (bul) at (7,2.75){$\bullet$};

\draw[dotted]  (6.5,3.5) -- (6.5,2);
\draw[dotted]  (6.8,3.5) -- (6.8,2);

\draw[->] (6.8,3.75) to [out=0, in=+90](7,2.85);

\draw[->, rounded corners] (7,2.75) --  (7.2,2.85) -- (7.4,2.65) --  (7.6,2.85) -- (7.8,2.65) -- (8.1,2.75);
\draw[->, rounded corners] (8.8,2.75) -- (8.9,2.85) -- (9.1,2.65) -- (9.3,2.75);

\draw (2,1.54) -- (2,-0.95);
\draw (2.3,1.48) -- (2.3,-1.02);

\draw[draw=black, fill=gray, opacity=0.2] (2,0) rectangle (2.3,0.5);
\node at (1.7,0.25){$q_i$};

\node at (2.15,0.25){$0$};

\draw[->] (8.25,2.5) -- (2.3,0.4);
\draw[->](2.3,0.1) -- (8.75, 2.5);

%\draw[->] (8.2,2.5) to [out=-90, in=0](2.3,0.4);
%\draw[->] (2.3,0.1) to [out=0, in=-90](8.75,2.5);

\node at (10.4,3.2){$c_i$};
\node (bot) at (10.4,2){column $k_i$};
\draw[->](10.4,2.2) -- (10.4, 2.4);

\draw[ rounded corners] (1,1.5) --  (1.5,1.45) -- (2,1.55) --  (2.5, 1.45) -- (3,1.55);
\draw[ rounded corners] (1,-1) --  (1.5,-1.05) -- (2,-0.95) --  (2.5, -1.05) -- (3,-0.95);

\draw[->] (9.3,2.75) -- (9.6,2.75);
\draw[dotted] (9.6,2.75) -- (9.95,2.75);
\draw[->] (9.95,2.75) -- (10.3,2.75);

\draw[ rounded corners] (1,1.5) --  (1.05,1) -- (0.95,0.5) --  (1.05, 0) -- (0.95,-0.5) -- (1,-1);
\draw[ rounded corners] (3,1.55) --  (3.05,1) -- (2.95,0.5) --  (3.05, 0) -- (2.95,-0.5) -- (3,-0.95);

\node (top) at (2.15, 2.65){column $d_i$};
\draw[->] (top) -- (2.15, 1.7);

\node at (1.5,-0.35){\Large $V$};

\draw[dotted] (7.8, 2.5) -- (7.8, -1.5);
\draw[dotted] (9.3, 2.5) -- (9.3, -1.5);
\draw[dotted] (10.8, 2.5) -- (10.8, -1.5);

\draw [thick, decorate,,decoration={mirror, brace,amplitude=10pt},xshift=-4pt,yshift=0pt]
(7.95,-1.5) -- (9.45,-1.5);

\node at (8.5, -2.1){left half};
\node at (8.5, -2.5){ of $\mathcal{M}(d_i)$};

\node at (10.1, -2.1){right half};
\node at (10.1, -2.5){ of $\mathcal{M}(d_i)$};

\draw [thick, decorate,,decoration={mirror, brace,amplitude=10pt},xshift=-4pt,yshift=0pt]
(9.45,-1.5) -- (10.95,-1.5);

\draw[->, dotted] (10.5,2.75) to [out=0, in=-90](11,3.75);

\end{tikzpicture}
\caption{Construction of pointer chain $\mathbf{p}_i$} \label{stitch}
\end{figure}
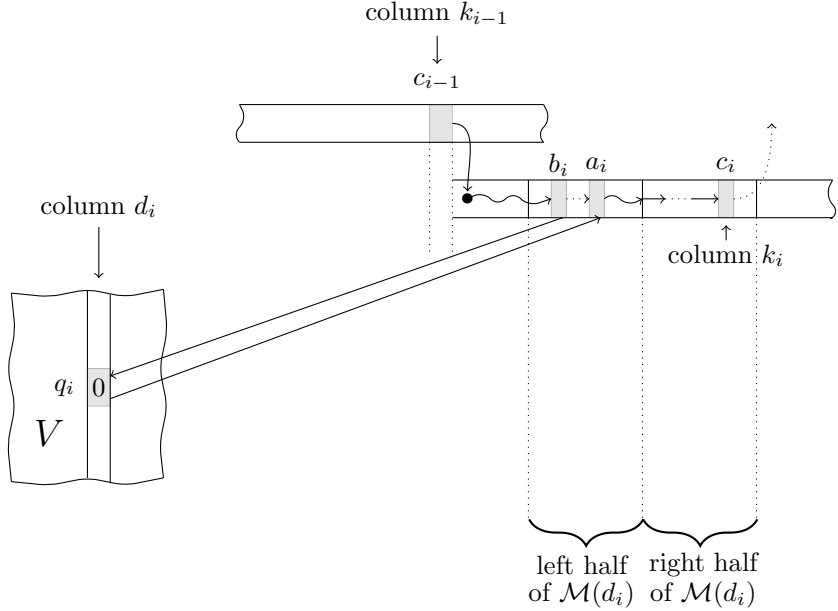

\section{Proof of the matching lemma}

We will show that every subset $R \subseteq V'$ of at most
$s/(\sqrt{r}(\log s)^4)$ columns has at least $|R|$ neighbors in $W'$.
Then, the claim will follow from Hall's theorem.

Observe that with high probability every column in $V'$ has
$\Omega(\sqrt{r}\log s)$ pointers leaving it. We expect these pointers
to be uniformly distributed among the at most $\log s$ blocks in $W$;
in particular, we should expect that every column in $V'$ sends
$\Omega(\sqrt{r})$ pointers into each block. We now formally establish this.
\begin{claim} \label{cl:eachcolexpands}
Let $V_j$ be the $j$-th column of $V'$ and $W_{j'}$ the $j'$-th block of $W$; then,
\[
\Pr[\forall j,j': |\ptr(V_j) \cap W_{j'}| \leq 400 \sqrt{r}] =o(1).\]
\end{claim}
\begin{proof}
Fix a location in $\ell \in V_j$. Let $\chi_\ell$ be the indicator
variable for the event $\ptr_\ell \in W_{j'}$. Then, the number of
pointers from $V_j$ into $W_{j'}$ is precisely $\sum_{\ell \in V_j}
\chi_\ell$. Since
\[ \Pr[\chi_\ell = 1] \geq \frac{500 \log s}{\sqrt{r}} \times \frac{1}{\log s} = \frac{500}{\sqrt{r}},\]
the expected number of pointers from column $V_j$ into $W_{j'}$ is at
least $500 \sqrt{r}$. Our claim follows from the Chernoff bound and
the union bound (over choices of $j$ and $j'$ since $r = \Omega((\log
s)^3)$). Here, we use the following version of the Chernoff bound (see
Dubhashi and Panconesi~\cite{ DBLP:books/daglib/0025902}, page 6): for the
sum of $r$ independent 0-1 random variables $Z_\ell$, each taking the
value $1$ with probability at least $\alpha$,
\[\Pr[ \sum_\ell X_\ell \leq (1-\eps)\alpha r ] \leq \exp(- \frac{\eps^2}{2} \alpha r)
.\]
Note that in our application $\alpha r \gg \sqrt{r} \geq \log s$.
\end{proof}

Suppose $j$ is such that $2^j \leq |R| < 2^{j+1}$. Then, we will show
that $R$ has the required number of neighbors among the bands of
the block $W_j$.
\begin{claim} \label{cl:manyheavybands}
For a set $R \subseteq V'$ and a block $W_j$, consider the set of bands of $W_j$
into which at least $2\sqrt{r}$ pointers from $R$ fall, that is,
\[ B_j(R) := \{b \in W_j: |\ptr(R) \cap b| \geq 2\sqrt{r}\},\]
 Then, for $j=1, \ldots, K$ and for all all  $R$ such that $2^j \leq |R|< 2^{j+1}$, we have
\[ \Pr[|B_j(R)| \leq 2|R|] =o(1).\]
\end{claim}
\begin{proof}
We will use the union bound over the choices of $j$ and $R$. Fix the
set $R$. We may, using Claim~\ref{cl:eachcolexpands}, condition on
the event that there are at least $400 \sqrt{r}|R|$ pointers from $R$
to $W_j$. Fix $400 \sqrt{r}|R|$ of these pointers. Now,  the number
of pointers that fall outside $B_j(R)$ is at most $20 \cdot
2^j \cdot 2 \sqrt{r} \leq 100 \sqrt{r}|R|$. That is,if $|B_j(R)| < 2|R|$, then there is a set
$T$ of $2|R|$ bands into which more than $400 \sqrt{r}|R|-100 \sqrt{r}|R|
=300 \sqrt{r}|R|$
pointers from $R$ fall. We will show that it is unlikely for such a
set $T$ to exist.  For a fixed $T$, the probability of this
event is at most
\[ {400 |R| \sqrt{r} \choose 300 |R| \sqrt{r}}
\left( \frac{2 |R|}{20 \cdot 2^j} \right)^{300 \sqrt{r}|R|} \leq 2^{-100
  \sqrt{r} |R|}.\] Using the union bound to account for all choices of
$R$ and the ${20 \cdot 2^j \choose 2|R|}$ choices of $T$, and using the fact that
$\sqrt{r} \gg \log s$, we conclude
that the probability that $B_j(R)$ fails to be large enough is at most
\[
\sum_{j=0}^{\log s - 3 \log \log s} \quad
\sum_{m=2^j}^{2^{j+1}-1}  {s/2 \choose m} 
{20 \cdot 2^j \choose 2m}
2^{-100 \sqrt{r}m} = o(1).\]
\end{proof}

In order to show that with high probability the set $R$ has the
required number of neighbors, we will condition on the high
probability event of Claim~\ref{cl:manyheavybands}, that is, $|B_j(R)|
> 2|R|$.  Let $\mathcal{B}$ be the set of such bands $b$ that receive
at least $2\sqrt{r}$ pointers. For each $b \in \mathcal{B}$, let
$P(b)$ be a set of $2 \sqrt{r}$ locations in the columns in $R$ whose
pointers land in $b$. If in at least $|R|$ of the $2|R|$ such bands
$b$, there is a pointer from $P(b)$ satisfying the conditions
(c1)--(c3), then we will have obtained the required expansion.  Fix a
pointer out of $P(b)$ (which by definition of $P(b)$ lands in band
$b$), and consider the following events.
\begin{description}
\item[$\event_1$:] The pointer leads to the same segment as a previous
  pointer (assume the locations in $P(b)$ are totally ordered in some
  way).
\item[$\event_2$:] The pointer leads to the first entry of the
  pointer chain in its segment (so, that location has no predecessor).
\item[$\event_3$:] At least $w_j/8$ entries of the segment that the pointer
lands in, are probed by $\cQ$.
\item[$\event_4$:] The
  predecessor of the location where the pointer lands is probed by
  $\cQ$.
\end{description}
Consider the pointers that emanate from $P(b)$ and land in some
band $b \in \mathcal{B}$.  Let
$n_1$ be the number of those pointers for whom $\event_1$ holds; let
$n_2$ be the number of those pointers for whom $\event_2$ holds; let
$n_3$ be the number of those pointers for whom $\event_3$ holds but
$\event_1$ does not hold; let $n_4$ be the number of those pointers for
whom $\event_4$ holds but $\event_1, \event_2$ and $\event_3$ do not hold.

If the claim of our lemma does not hold, then it must be that in at
least $|R|$ of the $2|R|$ bands of $\mathcal{B}$, all pointers that fall there fail to
satisfy at least one of the conditions (c1)--(c3); that is, one of $\event_1, \ldots,
\event_4$ holds for all $2\sqrt{r}$ of them. This implies that
\begin{equation}
n_1 + n_2 + n_3 + n_4 \geq 2\sqrt{r} |R|.
\end{equation}
To prove our claim, we will show that with high probability each
quantity on the left is less than $\sqrt{r}|R|/2$. In the following,
we fix a set $R$ and separately estimate the probability that one of
the quantities on the left is large. To establish the claim for all
$R$, we will use the union bound over $R$. In the proof, we use the
following version of the Chernoff-Hoeffding bound, which can be found in
Dubhashi and Panconesi (\cite{DBLP:books/daglib/0025902}, page $7$).
\begin{lemma}[Chernoff-Hoeffding bound]
\label{chernoff}
Let $X:=\sum_{i \in [n]} X_i$ where $X_i, i \in [n]$ are independently distributed in
 $[0,1]$. Let $t > 2e\mathbb{E}[X]$. Then
\[\mathbb{P}[X > t] \leq 2^{-t}.\]
\end{lemma}
\begin{claim}
$\displaystyle \Pr[ n_1 \geq \sqrt{r}|R|/2] \leq 2^{-r|R|/2}.$
\end{claim}
\begin{proof}
The probability that a pointer from $P(b)$ falls on a segment of a
previous pointer is at most $2\sqrt{r}/r$. Thus, the expected value of
$n_1$ is at most $8|R|$. We may invoke lemma~\ref{chernoff} and conclude that
\[ \Pr[ n_1 \geq \sqrt{r}|R|/2] \leq 2^{-\sqrt{r}|R|/2}. \]
\end{proof}

\begin{claim}
$\displaystyle \Pr[ n_2 \geq \sqrt{r}|R|/2] \leq 2^{-\sqrt{r}|R|/2}.$
\end{claim}
\begin{proof}
A pointer falls on head of random pointer chain in 
a segment with probability at most $2/w_j$. Thus,
\begin{align}
\E[n_2] & \leq \left(\frac{2}{w_j}\right) 4 \sqrt{r} |R| \leq \frac{160 |R|}{(\log s)^3} \nonumber.
\end{align}
Again, our claim follows by a routine application of Lemma~\ref{chernoff}.
\end{proof}
\begin{claim} $\displaystyle \Pr[ n_3 \geq \sqrt{r}|R|/2] = 0$.
\end{claim}
\begin{proof}
If $n_3 \geq \sqrt{r}|R|/2$, then the total number of locations read by
$\cQ$ is at least
\begin{align*}
n_3 \frac{w_j}{8} & \geq \left(\frac{\sqrt{r}|R|}{2}\right)\cdot  \frac{w_j}{8} \\
        & \geq \left(\frac{\sqrt{r} 2^{j}}{2}\right)
                \left(\frac{s}{8 \cdot 20 \cdot 2^{j}\log s}\right)  \\
        & \gg  \frac{\sqrt{r}s}{320 \log s}.
\end{align*}
This contradicts our assumption that $\cQ$ makes at most $\sqrt{r}s/
(\log s)^4$ queries.
\end{proof}

\begin{claim}
$\displaystyle \Pr[ n_4 \geq \sqrt{r}|R|/2] \leq 2^{-r |R|/2}.$
\end{claim}
\begin{proof}
Let us first sketch informally why we do not expect $n_4$ to be
large. Recall that in our random input we place a random pointer chain
in the left half of each segment. Once a pointer has landed at a
location in this segment, its predecessor is equally likely to be any
of the other locations in the segment. So the first probe into that
segment has probability about one in $w_j/2-1$ of landing on the
predecessor, the second probe has probability about one in $w_j/2 -2$
of landing on the predecessor, and so on. Since we assume $\event_3$
does not hold, there are at least $w_j/2-w_j/8-1$ possibilities for
the location of the predecessor. This implies that in order
for $n_4$ to be at least $\sqrt{r}|R|/2$ the query algorithm $\cQ$
must make $\Omega(w_j \sqrt{r}|R|/2)$ queries; but this number exceeds
the number of probes $\cQ$ is permitted.

In order to formalize this intuition, fix (condition on) a choice of
pointers from $V$. Let us assume that the algorithm makes $t$
probes. For $i =1,2,\ldots,t$, define indicator random variables
$\chi_i$ as follows: $\chi_i =1$ iff the following conditions hold.
\begin{itemize}
\item Suppose the $i$-th probe is made to a segment $p$ in band $b \in
  \mathcal{B}$ .  Let $\ell$ be the location where the first pointer
  (among the pointers from $P(b)$ to $p$) lands. Then, the $i$-th
  probe of $\cQ$ is made to the predecessor of $\ell$ in the
  random pointer chain in $b$.
\item Fewer than $w_j/8$ of the previous probes were made to this segment.
\end{itemize}
Observe that if more than one pointer land on $p$, then except
for the first amongst them (according to the ordering on the locations
in $P(b)$), event $\event_2$ does not hold for the remaining pointers,
and hence by definition  event $\event_4$ does not hold either.

Define $Z=\sum_{i=1}^t \chi_i.$. Note that $Z$ is an upper bound on $n_4$, and we wish to estimate the probability that $Z\geq \sqrt{r}|R|/2$.
The key observation is that for every choice $\sigma$ of
$\chi_1,\chi_2,\ldots, \chi_{i-1}$, we have
\begin{equation}\label{dom}
\Pr[\chi_i = 1 \mid
  \chi_1,\chi_2,\ldots, \chi_{i-1}= \sigma] \leq \frac{1}{3w_j/8-1} \leq \frac{4}{w_j}.
\end{equation}
Thus, 
\begin{align}
\E[Z] &\leq \left(\frac{4}{w_j}\right)t  \leq \left(\frac{4}{w_j}\right) (\log s)^{-4} \sqrt{r} s \leq (\log s)^{-2} \sqrt{r} |R| \nonumber.
      %&\leq \left(\frac{8\cdot 20\cdot 2^j \cdot \log s}{3s}\right)
          %  (\log s)^{-4} \sqrt{r} s\\
      %&\leq (\log s)^{-2} 2^j \sqrt{r}\\
      %& \leq (\log s)^{-2} \sqrt{r} |R|.
\end{align}
The variables $\chi_i$ are not independent, but it follows from~(\ref{dom}) that Lemma~\ref{chernoff}
is still applicable in this setting. We conclude that
\begin{equation}
\Pr[Z \geq \sqrt{r}|R|/2]  \leq 2^{- \sqrt{r} |R|/2}. \nonumber
\end{equation}
Since, the above bound holds for each choice of pointers from $V$, it
holds in general.
\end{proof}

Finally, to establish the required expansion for all sets $R$, 
we use the union bound over all $R$. The probability that some set
$R$ has fewer than $|R|$ neighbors is at most
\begin{align*}
&4 \sum_{k=1}^{s/(\sqrt{r} (\log s)^4)} {s/2 \choose k} 2^{-\sqrt{r} k/2} \\
& \leq \sum_{k\geq 1} s^k 2^{- \sqrt{r} k /2}\\
& \leq \sum_{k \geq 1} s^{-k} = o(1),
\end{align*}
where we used our assumption that $r \gg (\log s)^2$.
This completes the proof of the matching lemma.

%
%
%
%
%\section{Each block is good with high probability}
%Fix an algorithm $A$.
%\begin{definition}
%We say that $W_j$ is good for a set of columns $R$ if for each row $i \in [r]$, there is a pointer chain $a$
%such that
%\begin{itemize}
%\item $a$ includes all locations in $W_j$;
%\item $a$ covers at least $|R| - s/2^{j}$ columns of $R$;
%\item $\tail(a)$ is in row $i$.
%\end{itemize}
%We say that $W_j$ is good if it is good for every set of columns $R$
%such that $ s/2^{j} < |R| \leq s/2^{j-1}$.
%\end{definition}
%
%\begin{lemma}
%For $j=1,2,\ldots, \log s + 1$,
%$\Pr[\mbox{$W_j$ is good}] \geq \frac{1}{s}$.
%\end{lemma}

\paragraph{Acknowledgment:} We thank Sagnik Mukhopadhyay for useful discussions.
\bibliography{ref}

\end{document}